\documentclass[12pt]{amsart}

\usepackage{amsmath}
\usepackage{amssymb}
\usepackage{amsthm}
\usepackage{hyperref}
\usepackage{xypic}
\usepackage{graphicx}

\usepackage{fullpage}
\usepackage{setspace}
\usepackage{xcolor}

\theoremstyle{plain}
\newtheorem{definition}{Definition}
\newtheorem{proposition}{Proposition}

\newtheorem{corollary}{Corollary}

\theoremstyle{definition}

\theoremstyle{remark}

\DeclareMathOperator{\argmin}{argmin}
\def\shf{\mathcal}

\title{Superresolving star clusters with sheaves}

\author{Michael Robinson and Christopher Capraro}\thanks{SRC, Inc., (mrobinson@srcinc.com)}

\begin{document}

\begin{abstract}
This article explains an optimization-based approach for counting and localizing stars within a small cluster, 
based on photon counts in a focal plane array.
The array need not be arranged in any particular way, 
and relatively small numbers of photons are required in order to ensure convergence.
The stars can be located close to one another, as the location and brightness errors were found to be low when the separation was larger than $0.2$ Rayleigh radii.
To ensure generality of our approach, it was constructed as a special case of a general theory built upon topological signal processing using the mathematics of sheaves.
\end{abstract}

\maketitle



\section{Introduction}
\label{sec-intro}

This article explains an optimization-based approach for counting and localizing stars wthin a small cluster, 
based on photon counts in a focal plane array.
The array need not be arranged in any particular way, 
and relatively small numbers of photons are required in order to ensure convergence.
The stars can be located close to one another, and good performance was obtained when the separation was larger than $0.2$ Rayleigh radii.
To ensure generality of our approach, it was constructed as a special case of a general theory built upon topological signal processing using the mathematics of sheaves.
While familiarity with sheaves is assumed in Section \ref{sec-approach}, it is not essential to understand the algorithm derived using them that is discussed in Section \ref{sec-results}.

\subsection{Historical context}
\label{sec-history}

Perhaps the most famous algorithm for separating nearly coincident stars in the CLEAN algorithm and its generalizations (see \cite{Hogbom_1974,mckinnon1990spectral,Stewart_2011} for instance, among many others).
These greedy algorithms rely on the geometric structure of stars, namely that they are point sources or nearly so. 

While greedy algorithms can yield good performance, they also can fail in dramatic ways.
This can be especially pronounced if photon counts are low.
This article proposes that an optimization-based approach can lend some robustness to an algorithm.
Optimization-based imaging using the Wasserstein metric has recently led to improvements in molecular microscopy \cite{mazidi2020quantifying}.
Our approach is a generalization of this idea.


\section{Detailed problem statement}
\label{sec-problem_statement}

We can model a scene of $N$ stars illuminating a focal plane by using a simplistic Dirac point-mass model,
\begin{equation}
\label{eq:general_signal_model}
s(x; a_1, b_1, a_2, b_2, \dotsc, a_N, b_n) = \sum_{n=1}^N a_n \delta(|x-b_n|),
\end{equation}
where $a_n\in[0,\infty)$ is the brightness of each star, 
$b_n \in \mathbb{R}^2$ is the location of each star in the focal plane, 
and $x \in \mathbb{R}^2$ can be taken to be the location of an arbitrary pixel.

The response of an actual sensor to the incoming photons from a set of stars in the scene will not be a sum of Dirac point-masses
because of dispersion, diffraction, and other physical effects.  
Suppose that there are $M$ pixels in the focal plane, each collecting an integer number of photons.
This means that the measurements lie in $\mathbb{Z}^M$ (rather than $\mathbb{R}^M$).  
It is useful to record the photon counts as an integer-valued function $z=z(x) \in \ell^2(\mathbb{R}^2)$.
The problem to be solved in this article is therefore the determination of $N$, and the sets $\{a_n\}$ and $\{b_n\}$ from $z$.


\section{Algorithmic implementation and results}
\label{sec-results}

\label{sec-stellar-solution}

The algorithm we used to count, locate, and characterize stars begins with a proposed maximum number of stars $P$. 
Convergence is not ensured by Proposition \ref{prop:unknown_source_optimization} and Corollary \ref{cor:correct_sources} if the true number of stars $N$ is larger than $P$.
The algorithm proceeds through the following steps, which ultimately minimize the local consistency radii of the sheaf $\mathcal{J}_P$:
\begin{enumerate}
\item Collect photon counts from all pixels into a vector $z$
\item For $i$ in $1, 2, \dotsc, P$,
\begin{enumerate}
\item Initial guess: 
\begin{enumerate}
\item If $i=1$, use the centroid of the photon count distribution as the initial guess for $b_1$ and the total photon count as $a_1$
\item Otherwise, use the previous step's set of star locations $b_1, \dotsc, b_{i-1}$ and magnitudes $a_1, \dotsc, a_{i-1}$ as an initial guess with $a_i = 0$ and $b_i$ chosen randomly.
\end{enumerate}	
\item Iteration: Solve the minimization problem Equation \ref{eq:known_source_optimization} given $i$ stars using the $i$ centroids as initial guesses for $b_1, \dotsc b_i$.
\item Store: the residual error, along with the parameters for the $i$ stars.
\end{enumerate}
\item Return the proposed number of stars, locations, and brightnesses corresponding to the smallest residual.
\end{enumerate}
We found that the minimization problem Equation \ref{eq:known_source_optimization} was solvable using the standard Matlab {\tt fmincon} function, 
though the initial guess stage was necessary to aid in ensuring convergence under repeated Monte Carlo runs.

\begin{figure}
\begin{center}
\includegraphics[width=3in]{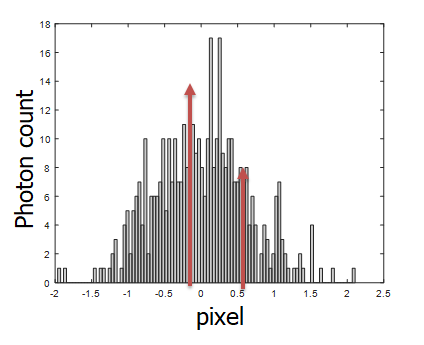}
\caption{A typical comparison between stars as Dirac distributions (arrows) and photon counts (histogram)}
\label{fig:star_metric}
\end{center}
\end{figure}

\begin{figure}
\begin{center}
\includegraphics[width=6in]{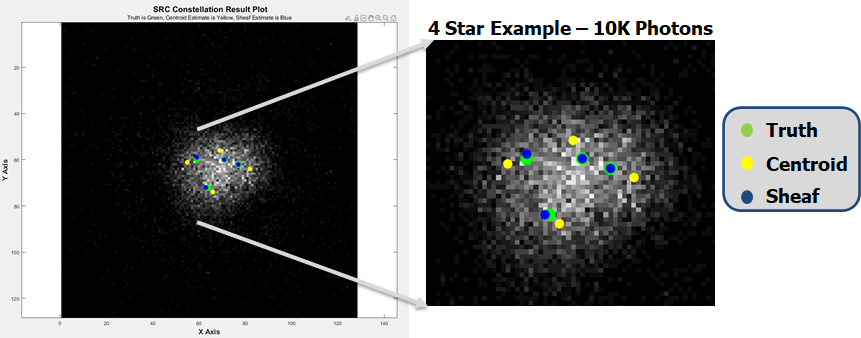}
\caption{A typical photon count distribution with $10 000$ photons (left) and (right), marked with the true star locations (green), a decomposition into four centroids (yellow), and the minimizer of consistency radius produced by our algorithm (blue)}
\label{fig:star_example}
\end{center}
\end{figure}

As a test of this approach, star fields consisting of between $1$ and $10$ stars were simulated.
These star fields were propagated through a simulated optical system, resulting in photon counts on an array of pixels.
These photon counts were then used to determine star locations and brightnesses using the sheaf approach discussed above,
a multiple centroiding approach, 
and the standard CLEAN algorithm.
Figure \ref{fig:star_example} shows a typical result for $10 000$ photons.

Assuming that the correct brightness and location of star $i$ is given by $(a_i,b_i)$, and that the algorithm's estimate of the same is $(a'_i,b'_i)$, 
an algorithm can be scored by its \emph{efficiency}, given by
\begin{equation}
\label{eq:efficiency}
Efficiency := 100 - \min_\phi \sqrt{\left(1-\frac{\#\text{image }\phi}{P}\right)^2+\sum_{i=1}^N \sum_{j=1}^P |a_{\phi(j)} - a'_i|^2+\sum_{i=1}^N \sum_{j=1}^P |b_{\phi(j)}-b'_i|^2},
\end{equation}
where $\phi : \{1, 2, \dotsc, P\} \to \{1,2, \dotsc, N\}$ ranges over all functions.
Note that $\#\text{image }\phi / P$ is the Jaccard index of the two sets of stars, given the function $\phi$ used as a matching.
The efficiency is closely related to a conical metric, but does not assume that the number of stars is the same in both sets.
Higher efficiency means that the algorithm is doing a better job of determining the correct number of stars, their brightnesses, and locations.
The best possible efficiency is $100$, while the worst is $0$.

\begin{figure}
\begin{center}
\includegraphics[width=7in]{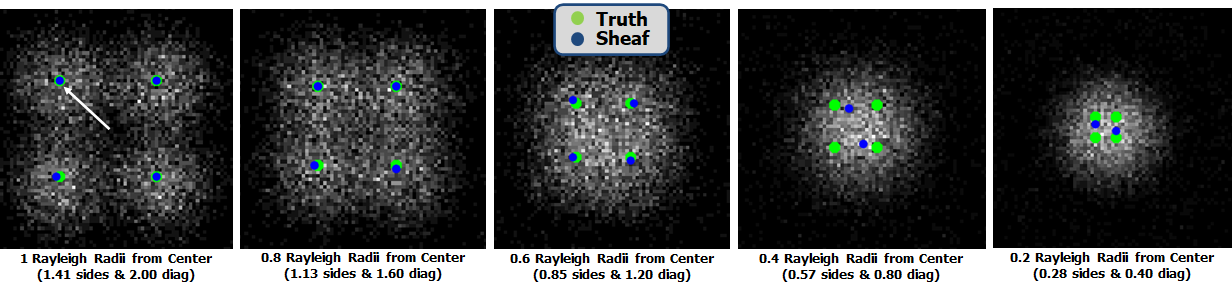}
\caption{Photon count distributions for four stars (green) at various inter-star separations, with the minimizers of consistency radius for the $\shf{J}_15$ sheaf superimposed (blue) and ten thousand ($10 000$) total photons..  (Star count unknown but less than $16$.)}
\label{fig:star_sepsweep_unknown}
\end{center}
\end{figure}

\begin{figure}
\begin{center}
\includegraphics[width=3in]{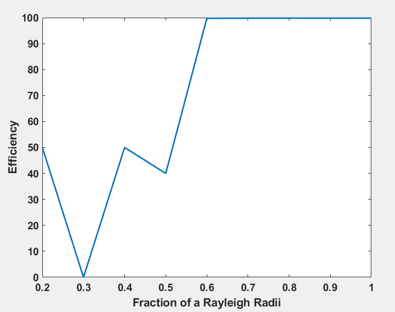}
\caption{Efficiency of the sheaf algorithm as a function of separation between two stars when the star count is unknown}
\label{fig:efficiency_v_sep_unknown}
\end{center}
\end{figure}

\begin{figure}
\begin{center}
\includegraphics[width=7in]{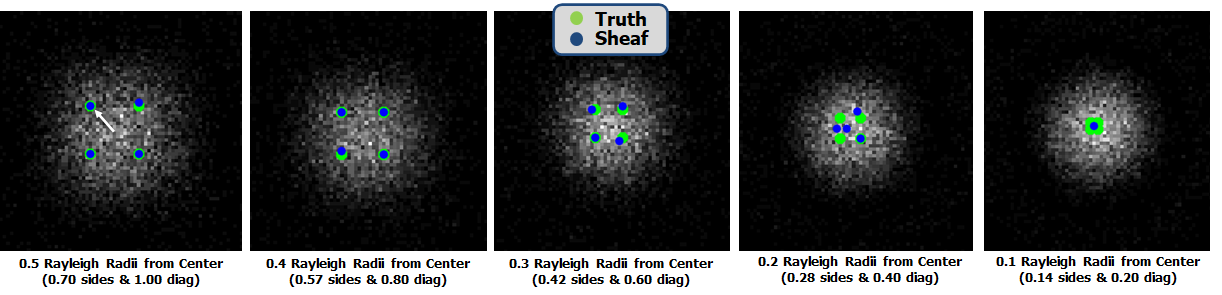}
\caption{Photon count distributions for four stars (green) at various inter-star separations, with the minimizers of consistency radius for the $\shf{J}_4$ sheaf superimposed (blue) and ten thousand ($10 000$) total photons..  (Star count known to be $4$.)}
\label{fig:star_sepsweep_known}
\end{center}
\end{figure}

\begin{figure}
\begin{center}
\includegraphics[width=3in]{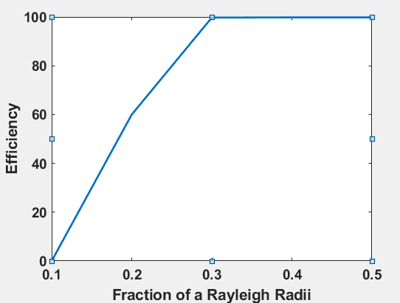}
\caption{Efficiency of the sheaf algorithm as a function of separation between two stars when the star count is known}
\label{fig:efficiency_v_sep_known}
\end{center}
\end{figure}

In order to compare the performance of the sheaf algorithm with the standard CLEAN algorithm \cite{Hogbom_1974}, we ran a set of $10$ Monte Carlo runs of $7$ different configurations of stars, as shown in Table \ref{tab:constellation_tests}.
Using the efficiency Equation \ref{eq:efficiency} to rate the performance of both algorithms,
the results from all runs of all tests are shown in Figure \ref{fig:clean_v_sheaf_overview}.
The sheaf algorithm did better than CLEAN in $55$ of the $70$ runs.
In $42$ of those $55$ runs the difference in efficiency was greater than $10$.

\begin{table}
\begin{center}
\caption{Constellation tests performed}
\label{tab:constellation_tests}
\begin{tabular}{|l|c|c|}
\hline
Test & Star count & Brightness\\
\hline
\hline
A&1&N/A\\
\hline
B&2&Equal\\
\hline
C&2&Variable\\
\hline
D&3&Equal\\
\hline
E&3&Variable\\
\hline
F&7&Variable\\
\hline
G&10&Variable\\
\hline
\end{tabular}
\end{center}
\end{table}

\begin{figure}
\begin{center}
\includegraphics[width=3in]{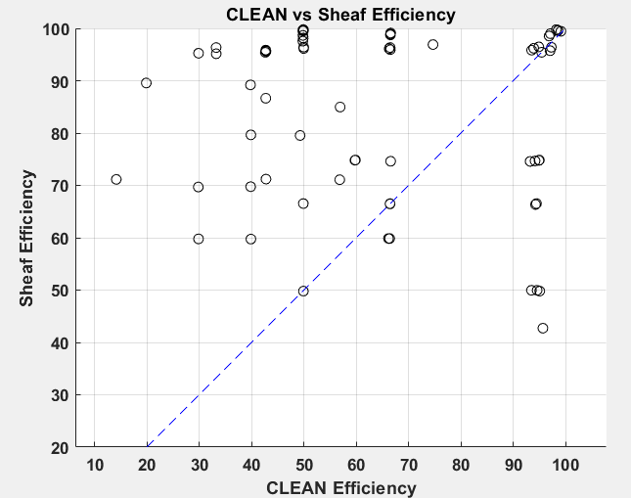}
\caption{Efficiency of the sheaf algorithm versus the CLEAN algorithm.  The dashed line is the diagonal (equal efficiency from both algorithms).}
\label{fig:clean_v_sheaf_overview}
\end{center}
\end{figure}

\begin{figure}
\begin{center}
\includegraphics[width=3in]{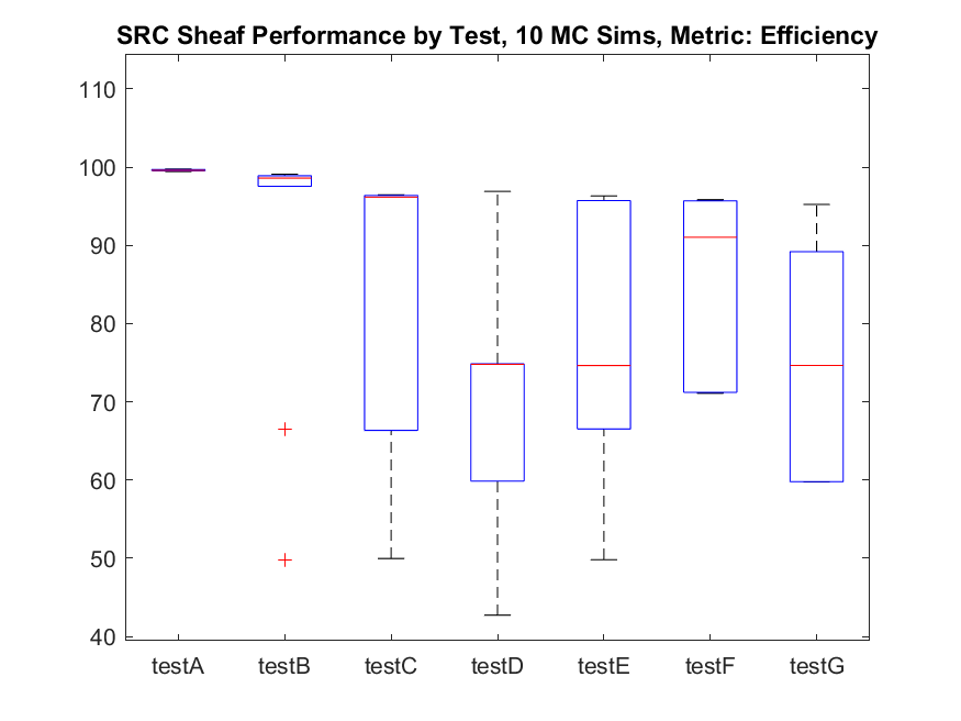}
\includegraphics[width=3in]{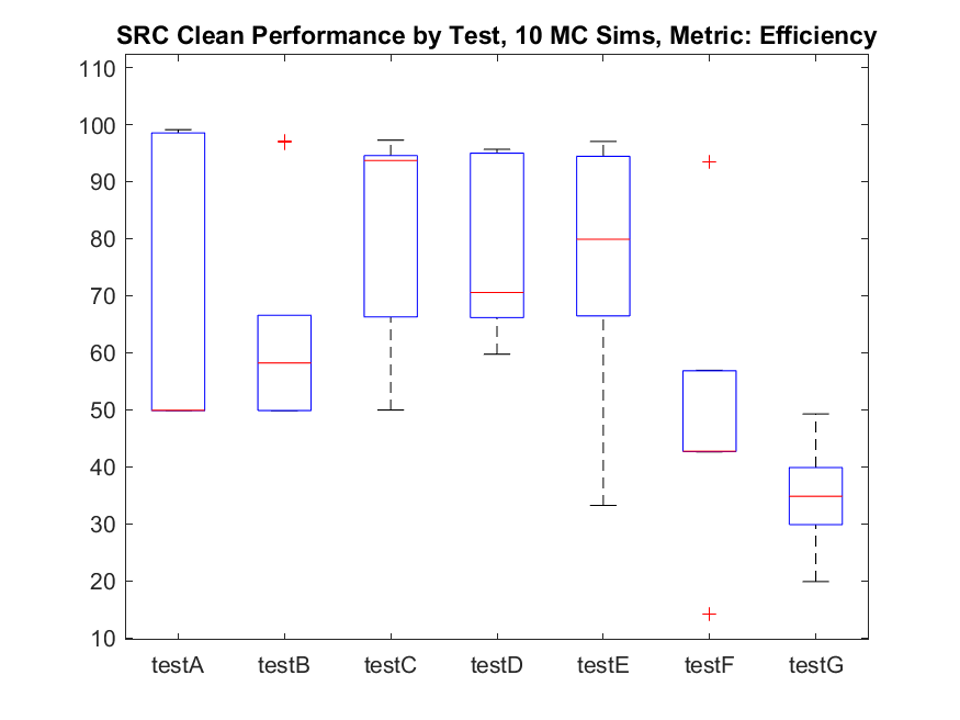}
\caption{Efficiency of the sheaf and CLEAN algorithms split by Test (see Table \ref{tab:constellation_tests}).}
\label{fig:clean_v_sheaf_eff}
\end{center}
\end{figure}

\begin{figure}
\begin{center}
\includegraphics[width=3in]{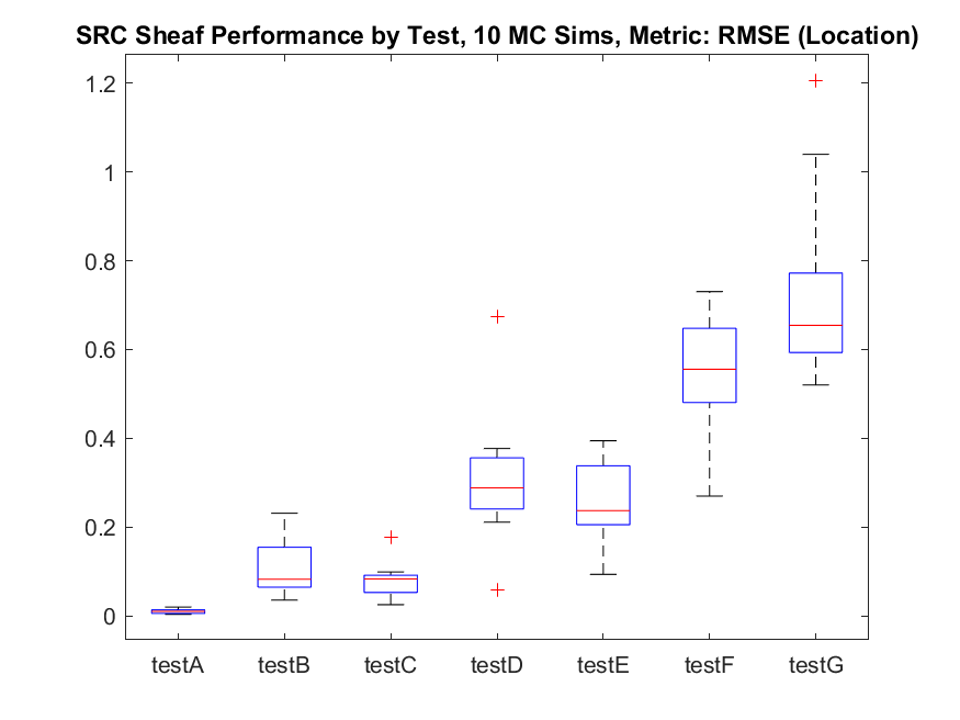}
\includegraphics[width=3in]{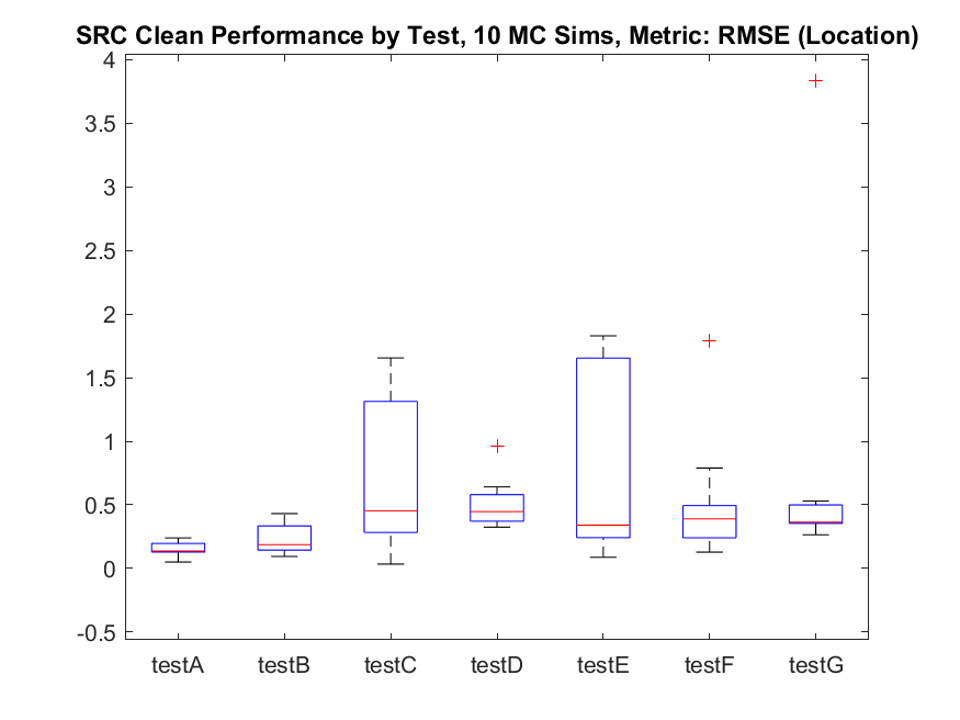}
\caption{Location root mean squared error (RMSE) of the sheaf and CLEAN algorithms split by Test (see Table \ref{tab:constellation_tests}).}
\label{fig:clean_v_sheaf_loc}
\end{center}
\end{figure}

\begin{figure}
\begin{center}
\includegraphics[width=3in]{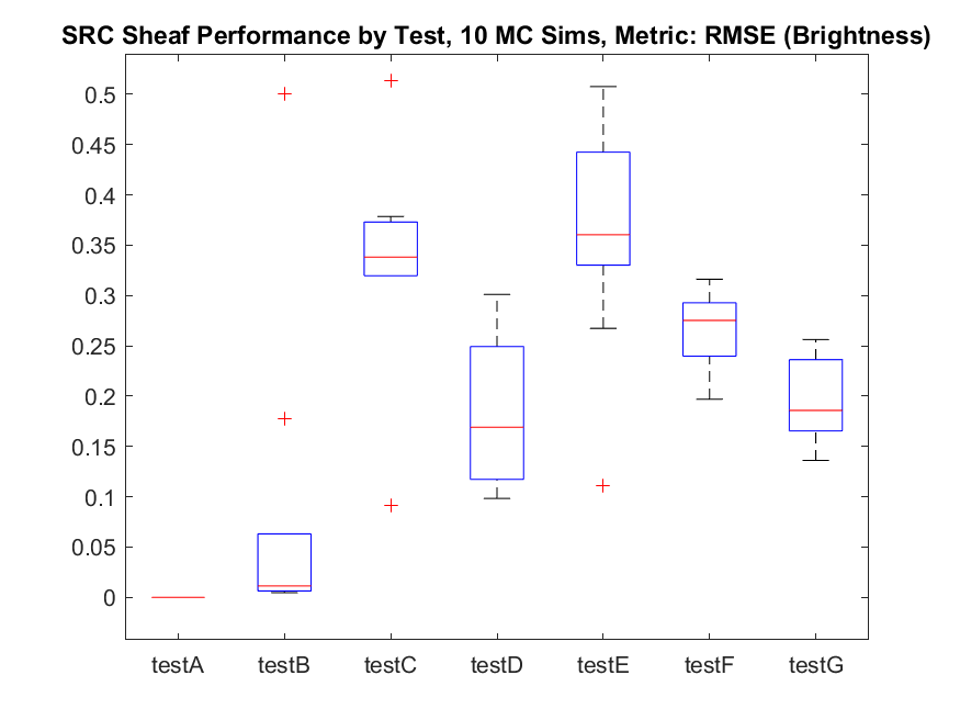}
\includegraphics[width=3in]{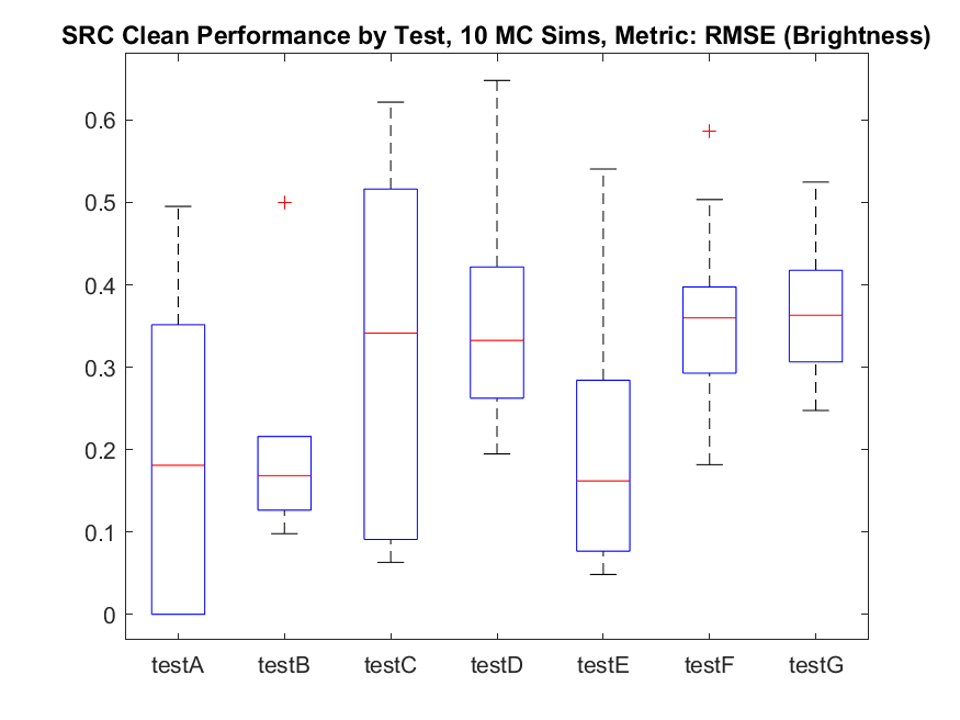}
\caption{Brightness root mean squared error (RMSE)  of the sheaf and CLEAN algorithms split by Test (see Table \ref{tab:constellation_tests}).}
\label{fig:clean_v_sheaf_mag}
\end{center}
\end{figure}

Delving more deeply into the performance differences, Figure \ref{fig:clean_v_sheaf_overview} shows a breakdown of efficiency for each Test shown in Table \ref{tab:constellation_tests}.  For Tests A and B the sheaf algorithm performs nearly optimally.  For Tests F and G (the hardest tests), the sheaf algorithm outperforms the CLEAN algorithm by a wide margin -- nearly twice the efficiency, as shown in Figure \ref{fig:clean_v_sheaf_eff}.  For the other tests, the sheaf algorithm's performance is substantially more variable than CLEAN.

The location errors for both algorithms and all tests are shown in Figure \ref{fig:clean_v_sheaf_loc}.  
As should be expected, location errors increase in both algorithms as the number of stars increase.  
Generally, the sheaf algorithm yields better location estimates than CLEAN for a smaller number of stars, 
while CLEAN yields better location estimates than the sheaf for a larger number of stars.

Brightness estimate errors for both algorithms and all tests are shown in Figure \ref{fig:clean_v_sheaf_mag}.  
Both tests yield highly variable performance on brightness estimates regardless of test conditions.
At least in the sheaf algorithm case, this may be a direct result of the use of the Euclidean metric on $\ell^2(\mathbb{R}^2)$,
since -- in contrast to the Wasserstein metric \cite{mazidi2020quantifying} -- the convergence in brightness is not guaranteed.  


\section{Conclusion}
\label{sec-conclusion}

The results we obtained from our algorithm typically exceeded the performance of the CLEAN algorithm, especially with more challenging scenes,
which indicates that this method shows promise.
The fact that it converges very well for small numbers of stars and less well for larger numbers of stars suggests there is room for improvement.
In particular, the use of a general optimization solver is not particularly efficient, and the solver converges slowly.

Computational considerations were not central to our approach on this problem, but could be important in future work.  
We have already developed a general library for sheaf algorithms \cite{pysheaf},
and could port this library to a high performance computing platform.
We expect that high-performance sheaf computational tools would have broad applicability, 
providing numerous avenues for improvements.

Finally, the sheaf-based approach discussed in Section \ref{sec-approach} is extremely general,
and applies without change to many other imaging techniques, not just focal plane arrays.
Quantum imaging with small number of photons \cite{tsang2019resolving} fits neatly into this perspective and deserves to be explored in future work.

\section*{Acknowledgements}
The authors thank Tom Ruekgauer and his team for the use of the simulated optical data presented in this article.

The views, opinions and/or findings expressed are those of the author and should not be interpreted as representing the official views or policies of the Department of Defense or the U.S. Government.
This material is based upon work supported by the
Defense Advanced Research Projects Agency (DARPA) under Agreement No. HR00112090125.


\section{Appendix: Modeling starfields sheafily}
\label{sec-approach}

Consider that each star is determined by its magnitude (an element of $[0,\infty)$) and location (an element of $\mathbb{R}^2$).
These properties have different units and thereby suggest that the correct space for the parameters defining a star is
\begin{equation*}
I = [0,\infty) \times \mathbb{R}^2 \subset \mathbb{R} \times \mathbb{R}.
\end{equation*}
Conventionally, we use $a$ for the brightness (or magnitude) of the star, and will use $b$ for its location.

If there are $N$ stars present in the scene, this would suggest that $I^N$ is the correct space of parameters defining all stars.
Since the order in which the stars are listed in such a product does not matter,
if there are $N$ stars present in the scene, the correct space of parameters for defining all of the stars is the quotient space $I^N/S_N$,
where $S_N$ is the group of permutations on $N$ items.

\subsection{Sheaf model warmup: known star count}
\label{sec-sheaf_known}

If we know the number of stars is $N$, we can express the relationship between the observations and the star parameters by a sheaf diagram
\begin{equation}
\label{eq:fixed_source_sheaf}
\xymatrix{
&I^N/S_N \ar[dl]_{s(x_1;\cdot)} \ar[d]^{s(x_2;\cdot)} \ar[drr]^{s(x_M;\cdot)} \\\
\mathbb{R} & \mathbb{R} & \dotsb & \mathbb{R}
}
\end{equation}
The diagram expresses the fact that each observation is determined by Equation \ref{eq:general_signal_model}, 
and that no observation functionally determines any other.
One should be aware that the lack of functional dependence does not imply that the observations are independent.
Specifically, it may happen that one or more measurements (values in the second row of the diagram) suffice to determine a unique value in $I^N/S_N$,
after which point all of the other measurements can be determined from this unique value.
This is exactly the situation in which an adaptive method is warranted, since later measurements are not actually required!

Suppose we have a set of $M$ measurements at $x_1$, $x_2$, $\dotsc$, $x_M$, 
for which we have observations $z_1$, $z_2$, $\dotsc$, $z_M$ respectively.
These observations form a partial assignment to the above sheaf, which is supported on the second row.

The partial assignment is not supported at the top row, since that would be tantamount to knowing all of the star parameters.
Nevertheless, assuming that no noise is present and that we have correctly found the parameters $a_1$, $b_1$, $\dotsc$, $a_N$, $b_N$, 
then it should be the case that
\begin{equation*}
z_m = s(x_m;a_1,b_1,\dotsc,a_N,b_N)
\end{equation*}
for all $m=1, \dotsc, M$.  

If $M$ real measurements are taken from the scene determined by $N$ stars,
then the previous discussions led to a simple sheaf diagram
\begin{equation*}
\xymatrix{
I^N/S_N \ar[d]^f\\ \mathbb{R}^M
}
\end{equation*}
where 
\begin{equation}
\label{eq:S_function}
f_N(a_1,b_1,\dotsc,a_N,b_N) := \begin{pmatrix}
s(x_1; a_1,b_1, \dotsc, a_N,b_N)\\
s(x_2; a_1,b_1, \dotsc, a_N,b_N)\\
\vdots\\
s(x_M; a_1,b_1, \dotsc, a_N,b_N)\\
\end{pmatrix},
\end{equation}
simply aggregates all of the pixel values into a single vector.
One can consider this as a function taking specific values on regions within the plane $\mathbb{R}^2$,
so we can interpret $f_N : I^N/S_N \to \ell^2(\mathbb{R}^2)$.

In the (unlikely) situation that the pixels have taken exactly the values specified by Equation \ref{eq:general_signal_model}, 
this corresponds to a global section of the sheaf diagram in Equation \ref{eq:fixed_source_sheaf}.
Since global sections have zero consistency radius and the stalks of the sheaf are all metric spaces, 
any assignment which is not a global section will have positive consistency radius.

On the other hand, if the $f_N$ function defined in Equation \ref{eq:S_function} is injective then there is only one global section.
This means that we can recast the problem of obtaining the stars from the observations as the minimization of consistency radius.
With essentially no further work, we obtain the following result.

\begin{corollary}
\label{cor:known_source_optimization}
If the function $f_N$ defined in Equation \ref{eq:S_function} is injective then the solution to
\begin{equation}
\label{eq:known_source_optimization}
\argmin_{\{a_n,b_n\}_{n=1}^N} \left(\sum_{m=1}^M \left|s(x_m;a_1,b_1, \dotsc, a_n,b_n) - z_m \right|^2\right).
\end{equation}
is the correct star decomposition for the scene $z_m$ if one exists.
\end{corollary}

The reader may correctly argue that the optimization problem in Equation \ref{eq:known_source_optimization} can be obtained easily -- though not solved -- by inspection!
However, if the number of stars is not known, 
then the correct optimization problem that obtains the stars parameters is difficult to state explicitly.
However, the correct optimization problem is still a consistency radius minimization,
but for a more elaborate sheaf.
This sheaf accounts both for the unknown number of stars and the fact that 
the number of measurements that should be taken depends on the (unknown) number of stars.

\subsection{Sheaf model encore: unknown star count}
\label{sec-sheaf_unknown}

Now let us assume that the scene is built from an unknown number of stars.
We cannot use a single $f_N$ map, because we do not know $N$.
Therefore, we need to consider many possibilities for the number of stars.
For instance, let us consider that there may be as many as $P$ stars in the scene. 

\begin{definition}
\label{def:j_p}
The sheaf $\shf{J}_P$ is defined by the diagram
\begin{equation}
\label{eq:j_p}
\xymatrix{
I\ar[drr]_{f_1(\cdot)} && I^2/S_2\ar[d]_{f_2(\cdot,\cdot)} && I^3/S_3 \ar[dll]_{f_3(\cdot,\cdot,\cdot)} && \dotsb & I^P/S_P \ar[dlllll] \\
&&\mathbb{R}^M\\
}
\end{equation}
in which each of the restriction maps are given by the map defined in Equation \ref{eq:S_function} with different numbers of stars.
\end{definition}

Notice that the $P$ in $\shf{J}_P$ is the maximum number of stars we will attempt to consider, 
which may not have much (if anything) to do with the actual number of stars $N$.
Of course, we hope that $N \le P$ so that the actual number of stars is considered in our analysis.
The global consistency radius of an assignment to this sheaf is given by
\begin{equation*}
\sum_{i=1}^P \|f_i(w_{i,1}, \dotsc, w_{i,i}) - z\|,
\end{equation*}
where $z \in \mathbb{R}^M$ is the measurement and the $w_{i,j}$ is the proposed $j$-th star parameter when we are considering a representation of the scene with $i$ stars.
The global consistency radius is simply the sum of consistency radii of each subproblem, in which we consider a given number of stars.
On the other hand, each one of these subproblems is the \emph{local} consistency radius for an open set in the base space topology for $\shf{J}_P$.
If the $f_\bullet$ functions are injective, then it is clear that if $N>P$, the only set upon which the local consistency radius can vanish is the set consisting of only the observation (ie. the bottom row of the sheaf diagram).
Conversely, if the $f_\bullet$ functions are injective and $N \le P$, then the minimum local consistency radius possible on the open set constructed as the star over the element with stalk $I^P/S_P$ is zero.

One particular drawback of $\shf{J}_P$ though, is that values of an assignment on the different stalks in the top row of Equation \ref{eq:j_p} have nothing to do with one another.
Another expression of this fact is that the base space topology for $\shf{J}_P$ is quite large.

\subsection{Pivoting to an algorithm}
\label{sec-probabilistic}

If we combine the sheaf methodology in the previous sections (Sections \ref{sec-sheaf_known} and \ref{sec-sheaf_unknown}) 
with observation that the collection of photon counts from each pixel are best thought of as a (single) function on the plane,
we obtain a way to recover star parameters from a distributional measurement.

To formalize this idea, consider the space $\ell^2(\mathbb{R}^2)$ of functions on star locations.
In brief, a measurement specifies the amount of mass on a set of star positions.  
The map $i_N: I^N/S_N \to \ell^2(\mathbb{R}^2)$ reinterprets stars as measures
\begin{equation*}
i_N\left((a_1,b_1), \dotsc (a_N,b_N)\right) := \sum_{n=1}^N a_n \delta_{b_n}(x).
\end{equation*}

We can use these $i_N$ maps as the restriction maps in $\shf{J}_P$, which still works structurally, namely
\begin{equation}
\xymatrix{
I\ar[drr]_{i_1(\cdot)} && I^2/S_2\ar[d]_{i_2(\cdot,\cdot)} && I^3/S_3 \ar[dll]_{i_3(\cdot,\cdot,\cdot)} && \dotsb & I^P/S_P \ar[dlllll]^{i_P} \\
&&\ell^2(\mathbb{R}^2)\\
}
\end{equation}

Algorithmically, if we obtain a measurement as an element $z\in \ell^2(\mathbb{R}^2)$ (which need not be a sum of a small number of Dirac distributions), then the best star decomposition is still obtained by minimizing local consistency radii:
\begin{equation*}
\|i_P(w_{1}, \dotsc, w_{P}) - z\|^2,
\end{equation*}
and selecting the smallest $P$ for which this quantity is zero.
In cases where the source magnitudes are widely varying, it is likely that a greedy approach is possible (and may be preferable).
In Section \ref{sec-results}, we employ a compromise approach by reusing solutions with smaller number of proposed stars when solving for larger numbers of proposed stars.

\begin{proposition}
\label{prop:unknown_source_optimization}
Suppose that the number of sources is fixed at $N$ and that $z \in \mathbb{R}^M$ is given by
\begin{equation*}
z= S(a_1,b_1, \dotsc, a_N,b_N),
\end{equation*}
where $S$ is given by Equation \ref{eq:S_function}.

If we assign $z \in \mathbb{R}^M$ to represent the measurement in the sheaf $\shf{J}_P$ (see Equation \ref{eq:j_p})  then there is an extension to a global assignment in which the local consistency radius vanishes for all sufficiently small open sets provided that $P \ge N$.
\end{proposition}
\begin{proof}
To that end, $z\in \mathbb{R}^M$ is assigned in the bottom row of the diagram in Equation \ref{eq:j_p}.

Let us name the values in the assignment to the top row of the sheaf diagram 
\begin{equation*}
\{(a_{1,1},b_{1,1})\}, \{(a_{2,1},b_{2,1}),(a_{2,2},b_{2,2})\}, \{(a_{3,1},b_{3,1}),(a_{3,2},b_{3,2}),(a_{3,3},b_{3,3})\}, \dotsc, \{(a_{P,1},b_{P,1}),\dotsc,(a_{P,P},b_{P,P})\}.
\end{equation*}
Notice that these values define many possibilities for the source parameters.

If the number of sources is not known, but decompositions with $P \ge N$ sources are encoded in the sheaf, 
then the consistency radius of any assignment supported on the star open set (minimal open set in the sense of inclusion) whose stalk is $I^N/S_N$ is given by
\begin{equation}
\label{eq:consistency_radius_fixed_measurements_windowed}
\sum_{i=N}^P \|S(a_{i,1},b_{i,1} \dotsc, a_{i,i},b_{i,i}) - z\|.
\end{equation}

Define
\begin{equation*}
a_{P,k} = \begin{cases}
a_k & \text{if } k \le N,\\
0 & \text{otherwise}
\end{cases}
\end{equation*}
and
\begin{equation*}
b_{P,k} = \begin{cases}
b_k & \text{if } k \le N,\\
\text{arbitrary} & \text{otherwise.}
\end{cases}
\end{equation*}

If we then declare that $a_{i,k} =a_{j,k}$ and $b_{i,k} =b_{j,k}$ whenever they are both defined, then this definition clearly results in the sum of Equation \ref{eq:consistency_radius_fixed_measurements_windowed} being zero.
\end{proof}

\begin{corollary}
\label{cor:correct_sources}
If the $S$ function is injective, then the assignment for which the local consistency radius of $\shf{J}_P$ vanishes occurs precisely at the number of sources $N$.
\end{corollary}

Therefore, from a methodological perspective, the number of sources can be determined from the consistency filtration given enough samples.
If there is noise, look for the ``knee'' in local consistency radius as the open set size increases.  
This is precisely our algorithmic approach.


\bibliographystyle{plain}
\bibliography{starsep-master}

\end{document}